\documentclass[journal,twocolumn]{IEEEtran}
\pdfoutput=1
\usepackage{cite}
\usepackage{amsmath,amssymb,amsfonts}
\usepackage{algorithmic}
\usepackage{graphicx}
\usepackage[usenames]{color}
\usepackage{balance}
\usepackage{textcomp}
\usepackage{verbatim}
\usepackage[overload]{textcase}
\usepackage[]{footmisc}

\usepackage{adjustbox}  
\usepackage{bbm}
\usepackage{amsthm}

\makeatletter

\newtheorem{prop}{Proposition}

\theoremstyle{definition}

\newtheorem{rem}{Remark}

\newcommand{\comm}[1]{}

\newcommand\numberthis{\addtocounter{equation}{1}\tag{\theequation}}

\newcommand{\Rmnum}[1]{\expandafter\@slowromancap\romannumeral #1@}

\newcommand{\FTP}{\mathsf{FTP}}
\newcommand{\ATP}{\mathsf{ATP}}

\def\BibTeX{{\rm B\kern-.05em{\sc i\kern-.025em b}\kern-.08em
T\kern-.1667em\lower.7ex\hbox{E}\kern-.125emX}}
\makeatother

\begin{document}

{\title{ Analysis of D2D Communication with RF Energy Harvesting and Interference Management}}

\author {
\IEEEauthorblockN{
Nasrin Razmi, \IEEEmembership {Student Member,~IEEE},
Mehdi Mahdavi,
Mohammadali Mohammadi, \IEEEmembership {Member,~IEEE},
\\Petar Popovski, \IEEEmembership{Fellow,~IEEE}
}

\thanks{N. Razmi and M. Mahdavi are with the Department of Electrical and Computer Engineering, Isfahan University of Technology, Isfahan, Iran (e-mail: n.razmi@ec.iut.ac.ir; m\_mahdavi@cc.iut.ac.ir).}

\thanks{M. Mohammadi is with the Faculty of Engineering, Shahrekord 115, Iran (e-mail: m.a.mohammadi@sku.ac.ir).}

\thanks{ P. Popovski is with the Department of Electronic Systems, Aalborg University, Denmark (e-mail: petarp@es.aau.dk).}
}
\maketitle

\begin{abstract}
\comm{12}
Device-to-device (D2D) underlaid cellular network,
enabled with radio frequency energy harvesting (RFEH),
and enhanced interference management schemes is a promising
candidate to improve spectral and energy efficiency of next generation wireless networks. In this paper,
we propose a time division duplexing (TDD)-based protocol, in which allows the devices to harvest energy from the downlink transmissions of the base station, while controlling the interference among D2D and cellular communication in the uplink. We propose two schemes for transmission coordination, based on fixed transmission probability (FTP)
and adaptive transmission probability (ATP), respectively. In FTP, the D2D transmitters that have harvested enough energy can initiate data transmission with a fixed probability. Differently from this, in ATP a device utilizes its
sensing capability to get improved coordination and interference control among the transmitting devices. 
We evaluate the network performance by presenting an accurate energy model and leveraging tools from stochastic geometry. The results on outage probability and D2D sum-rate reveal the importance
of transmission coordination on network performance. These observations led to a solution for choosing the parameters of the ATP scheme that achieves an optimal tradeoff between 
the D2D outage probability and number of transmitting users.

\hspace{1ex}

\begin{IEEEkeywords}
Device-to-device communication, cellular network, radio frequency energy harvesting, interference management, outage probability, stochastic geometry.
\end{IEEEkeywords}

\end{abstract}

\section{INTRODUCTION}

Device-to-device~(D2D) communication as an underlay to a cellular network refers to the direct communication between proximate users without going through the base station~(BS).
This feature can improve spectral and energy efficiency, delay and overall throughput \cite{6807945,6163598,8340813}.

However, D2D communication that uses the same spectrum as the cellular network may cause a significant interference, which needs to be dealt with through interference management schemes~\cite{8340813,5910123, 6047553,6736746, 8643973}. 

Another major challenge in future wireless networks is energy consumption~\cite{8766143}, which can be addressed through energy harvesting~(EH), leading eventually to devices that are self-powered~\cite{8869795,sakr2015cognitive}. Radio frequency energy harvesting~(RFEH) enables transceivers to restore energy by converting the received RF signals to electricity \cite{6951347}. RFEH 
is becoming more relevant due to the steady increase of electromagnetic waves in both indoor and outdoor environments at all times \cite{8664000,ulukus2015energy,6575083}. 
Although current harvesting circuits can only afford to save a limited amount of energy, they can be still suitable in a D2D setting due to the low power used for D2D transmissions. 

This paper treats D2D underlaid cellular network that uses RFEH, introduces novel communication schemes and provides a comprehensive analysis based on stochastic geometry.

\subsection{Related Work}
 Stochastic geometry has been widely used to model and analyze the interference of wireless networks \cite{6042301} and it has also been applied to performance analysis of EH-based D2D communication networks~\cite{kusaladharma2017performance,sakr2015cognitive,8069034,7782752,yang2016heterogeneous,7752450}. The authors in \cite{sakr2015cognitive}, studied the performance of EH-based cognitive D2D underlying multi-channel cellular communication, where D2D transmission was considered successful if both the EH and transmission process were successful. 
 The performance of a RFEH-based D2D network which had four different models for EH and transmission was addressed in \cite{kusaladharma2017performance} using Markov chain.
  The probability of harvesting enough energy for D2D users was evaluated in \cite{8069034} where D2D users were able to harvest energy from both BS and power beacons. In \cite{7782752}, a trade-off between the number of D2D transmissions and the amount of harvested energy was formed for EH-based D2D users which had only access to a portion of the cellular spectrum. Authors in \cite{yang2016heterogeneous}, studied the D2D relaying
for EH-based communications. In \cite{7752450}, the energy efficiency of D2D communication underlying multiple-input multiple-output cellular communication with EH from the dedicated power beacons and cellular users' transmissions was evaluated.

Resource allocation in EH-based D2D communication networks has been treated in \cite{8869795,8353158,ding2016dynamic,8536466,7872433}.
The work \cite{8869795} introduced resource management based on deep learning to maximize the sum
rate by controlling the transmission
power and the power splitting ratio. Maximizing the sum throughput using time scheduling and power control was provided in \cite{8353158}. The authors in \cite{ding2016dynamic} minimized the total energy cost of all D2D transmitters and cellular users, while guaranteeing the quality-of-service requirement for D2D and cellular communication by dynamic spectrum allocation. Energy efficiency maximization was also evaluated in \cite{8536466} using game-theoretic learning approach. In this context, joint spectrum resource allocation and power control problems were studied in \cite{7872433} in a simultaneous wireless information and power transfer based D2D network.

\subsection{Our Contribution}
In this paper, we design uplink~(UL) and downlink~(DL) of a time division duplexing~(TDD) protocol in order to support both interference management and RFEH. In this system, in the DL, the BS transmits to the cellular user, while the D2D transmitters remain idle and harvest energy from BS signal. During the UL, a cellular user transmits to the BS, while the D2D transmitters that have sufficient energy can communicate with their corresponding receivers. The main contributions of this work are summarized as follows:
\begin{itemize}
\item{We investigate RFEH-based D2D network under two different assumptions for the D2D transmitters: with and without sensing capabilities, respectively. In this regards, to coordinate the D2D transmissions occurring during the UL period, we propose two different schemes: fixed transmission probability~(FTP) and adaptive transmission probability~(ATP). In FTP, a D2D transmitter with sufficient energy can initiate a transmission with a fixed probability. In ATP, a D2D transmitter senses the channel before transmission in order to reduce the induced interference towards the other D2D links.}
\item {We present an accurate energy model based on the available energy in the batteries of D2D transmitters. By applying this model, we evaluate the performance of the FTP and ATP schemes. }
\item {Using stochastic geometry, the BS and D2D outage probabilities and average achievable D2D sum-rate are derived, providing insights into the impact that the system parameters have on the performance. }
\end{itemize}

The rest of this paper is organized as follows: In Section~\ref{sec:syss}, the system model, including the network model as well as EH model, for the D2D underlaid cellular network is presented. The analytical expressions for D2D and BS outage probabilities and average achievable D2D sum-rate for both FTP and ATP schemes are derived in Section~\ref{sec:Schemes}. Numerical results are described in Section~\ref{sec:num}, followed by conclusions in Section~\ref{sec:conc}.

$Notation$: $\Pr(\cdot)$ denotes the probability; $\mathbb{E}\left[X\right]$ and $f_{X}(x)$ are the expected value and probability density function~(PDF) of random variable $X$, respectively; $\Gamma(a)=\int_0^\infty e^{-x} x^{a-1} dx$ denotes Gamma function defined in \cite[Eq. (6.1.1)]{abramowitz1965handbook}; and $arccos\left(y\right)$ denotes the inverse of cosine function at $y$.

\section{System Model}~\label{sec:syss}

\subsection{Network Model}
We consider a single-cell cellular network underlaid D2D communication consisting of a single BS, a cellular user, and multiple D2D pairs as shown in Fig. \ref{fig:Network model} like \cite{lee2015power, 7797486}. The BS is located at the center of the cell with radius $R$, and the cellular user is randomly located within the cell. Moreover, D2D transmitters, constructing set $\Phi_d$, are distributed according to a homogeneous Poisson Point Process~(PPP) with density $\lambda_d$. Each D2D receiver is placed at a distance of $r_d$ meters from its transmitter with a uniform random direction \cite{lee2015power, 7797486}.

Cellular and D2D communications operate based on the TDD protocol where the DL and UL sub-slots alternate. The time is slotted and each time slot is divided into UL and DL sub-slots, each of them of a duration $T$, see~Fig. \ref{fig:GeneralModelEH}.
In DL, BS transmits to the cellular user, while the D2D transmitters remain idle and only harvest energy from BS transmissions. In an UL sub-slot, the cellular user transmits to the BS, while the D2D transmitters either communicate with their corresponding receivers or remain idle. In the rest of the paper, and without loss of generality, the duration of each UL and DL sub-slot is normalized $T=1$, making it possible to treat power as equivalent to energy.

\begin{figure}
\centering
\includegraphics[width=0.60\columnwidth]{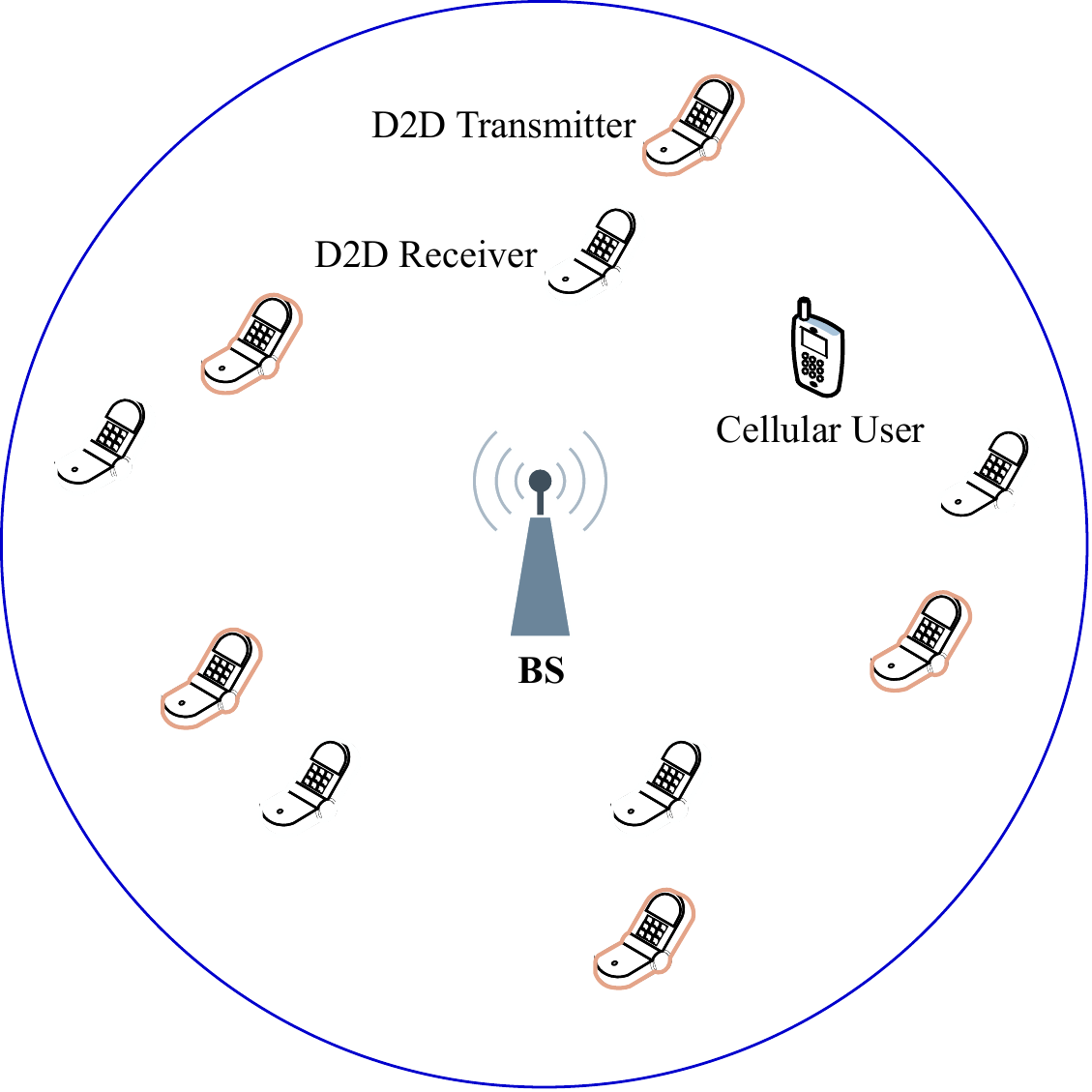}
\caption{Single-cell D2D underlaid cellular network consists of one BS, one cellular user and multiple randomly distributed D2D pairs.}
\label{fig:Network model}
\end{figure}

During the UL sub-slot, the received signal at BS can be expressed as

\begin{align*}
        y_{b}=d_{c,b}^{\frac{-\alpha}{2}}{h_{c,b}}s_c
        +\sum_{ j\in \phi_{t}} d_{j,b}^{\frac{-\alpha}{2}} {h_{j,b}}s_{j}+n_{b},
\numberthis
\label{received signal at BS}
\end{align*}

where $s_c$ is the signal sent by the cellular user and $s_j$ is the 
$j^{th}$ D2D transmitter interfering at the BS. The Rayleigh fading channel from $i^{th}$ D2D transmitter to $j^{th}$ D2D receiver is denoted by $h_{i,j}$ where the channel gain, $|h_{i,j}|^2$, follows the exponential distribution with unit average power. The distance and path loss between $i^{th}$ D2D transmitter and $j^{th}$ D2D receiver are denoted by $d_{i,j}$ and $d_{i,j}^{-{\alpha}}$, respectively where $\alpha$ stands for path loss exponent. Let $\phi_{t}$ denote a subset of all D2D transmitters i.e. $\Phi_d$ that are able to transmit based on the FTP/ATP schemes. Moreover, $n_{b}$ is the zero-mean additive white Gaussian noise~(AWGN) at the BS with noise power $N_0$.

In an UL sub-slot, the received signal by the $i^{th}$ D2D receiver can be expressed as

\begin{align*}
        y_{i} &=r_{d}^{\frac{-\alpha}{2}}{h_{i,i}}s_{i}+\sum_{j\in \phi_{t},{j}\neq{i}} d_{j,i}^{\frac{-\alpha}{2}}h_{j,i}s_{j}+d_{c,i}^{\frac{-\alpha}{2}}{h_{c,i}}s_c+n_{i},
\numberthis
\label{received signal at D2D}
\end{align*}

where $s_i$ is the signal from the desired, while $s_j$ from the undesired D2D transmitter. Moreover, $n_i$ is AWGN at $i^{th}$ D2D receiver with power $N_0$. Note that in \eqref{received signal at D2D}, the first term denotes the desired signal, while the second and third terms represent the interference from other D2D transmitters and the cellular user, respectively.

In the DL, the received signal by the cellular user from BS  can be expressed as

\begin{align*}
        y_{c}=d_{b,c}^{\frac{-\alpha}{2}}{h_{b,c}}s_{b}
        +n_{c},
\numberthis
\label{received signal at cellular user}
\end{align*}

 where $s_{b}$ denotes the transmitted signal by BS and $n_{c}$ is the AWGN at cellular user.
 
 By using \eqref{received signal at BS}, the signal-to-noise plus interference~(SINR) at the BS in the UL sub-slots, can be written as

\begin{equation}
        \Gamma_{b}=\frac{P_{c}{|h_{c,b}|^2}d_{c,b}^{-\alpha}}{\sum_{j\in \phi_{t}} P_d{|{h_{{{j},b}}}|^2}d_{{{j},b}}^{-\alpha}+N_0},
\label{sinr_bs_1}
\end{equation}

where $P_c$ and $P_d$ denote the transmission power of cellular user and D2D transmitters, respectively.

Moreover, by invoking \eqref{received signal at D2D}, the received SINR at $i^{th}$ D2D receiver in the UL sub-slots is given by

\begin{equation}
\begin{aligned}
        \Gamma_{i}\!\!=\!\!\frac{P_d{|h_{i,i}|^2}r_{d}^{-\alpha}}{\sum_{{j}\in \phi_{t} , \newline{j\neq i}} \!P_d|h_{j,i}|^2d_{j,i}^{-\alpha}\!+\!P_c{|h_{c,i}|^2}d_{c,i}^{-\alpha}\!+\!N_0}.
\label{eq:eq4}
\end{aligned}
\end{equation}

\subsection{Energy Harvesting Model}

Each D2D transmitter is equipped with a battery which extracts the energy of the RF signals transmitted by the BS in DL sub-slots using a power conversion circuit \cite{6951347}. Similar to \cite{yang2016heterogeneous,7782752}, infinite battery capacity is assumed for D2D transmitters to buffer the harvested energy.\footnote{Infinite battery capacity assumption helps to simplify the derived equations. However, the results can be extended to the finite capacity model.} Let $E_n^i$ denote the available energy in the battery of $i^{th}$  D2D transmitter at the beginning of $n^{th}$ UL sub-slot as shown in Fig. \ref{fig:GeneralModelEH}.
Based on this model, the available energy in the battery of the $i^{th}$ D2D transmitter can be expressed as \cite{6609136},

\begin{align*}
    E_{n}^{i}=E_{n-1}^{i}+H_{n-1}^{i}-P_d X_{n-1}^i,
    \numberthis
    \label{Eq: general energy model}
\end{align*}

where $H_{n-1}^{i}$ denotes the harvested energy by the $i^{th}$ D2D user at $(n-1)^{th}$ time-slot. $X_{n-1}^{i}$ gets values 0 or 1, according to the proposed transmission scheme in the next section.

\begin{figure}
\centering
\includegraphics[width=0.85\columnwidth]{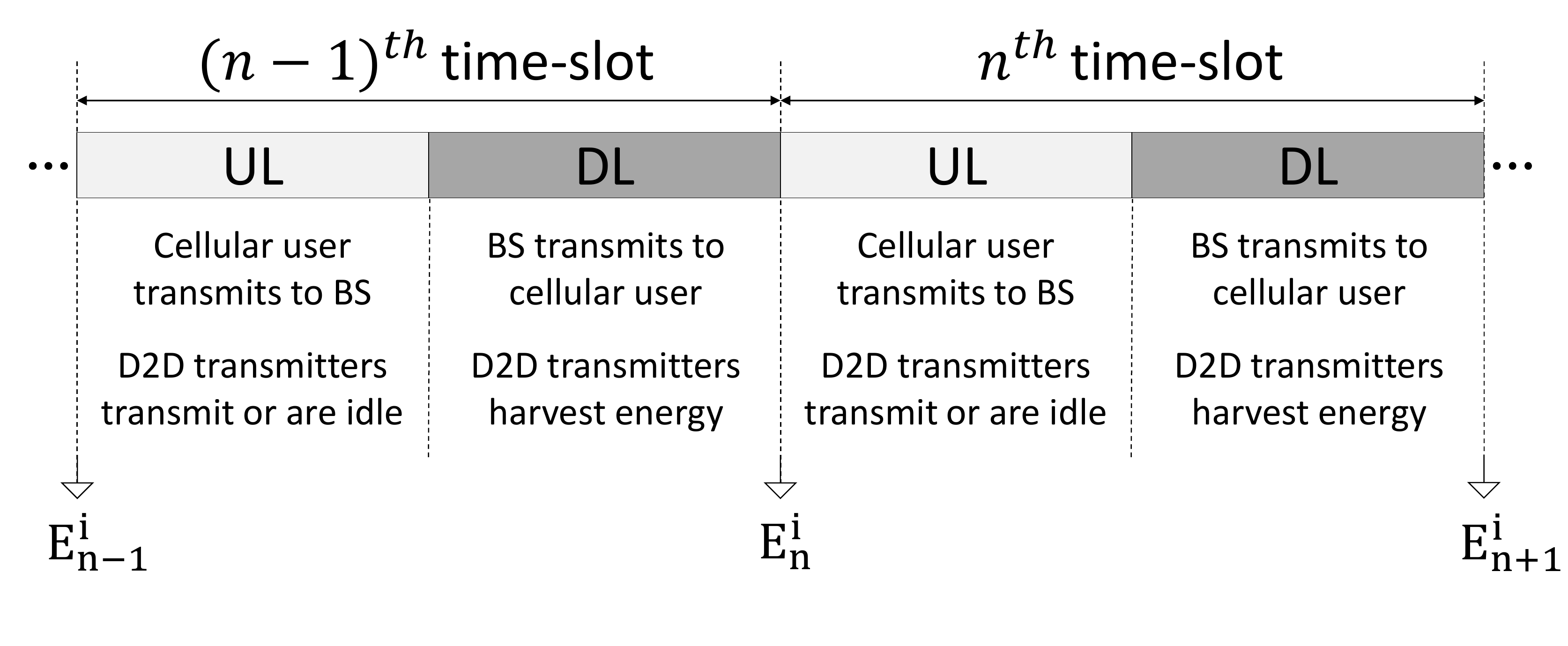}
\caption{Time-slots diagram, UL and DL placement and battery energy level for the $i^{th}$ user.}
\label{fig:GeneralModelEH}
\end{figure}

A D2D transmitter will be operable if its available energy level is greater than or equal to a predefined threshold, denoted by $E_{th}$. The battery of an operable user has the minimum energy necessary for transmission. The average number of D2D transmitters in the considered cell is $\lambda_d \pi R^2 $, where $\pi R^2$ is the cell area. By applying the thinning property of the Poisson process \cite{chiu1996stochastic}, the D2D transmitters which are operable constitute a PPP with density $\lambda_d\pi_{o}$, where $\pi_{o}$ denotes the probability of being operable, analyzed in the next section. The energy harvested in each DL sub-slot for the $i^{th}$ D2D transmitter which is located at the distance $d_{b,i}$ from BS is determined by $H_n^{i}=\eta P_{b} {|h_{b,i}|^2} d_{b,i}^{-\alpha}$ where $P_{b}$ is the transmission power of BS \cite{8333687}. Moreover, $0<\eta\leq1$ denotes the RF to DC power conversion efficiency~\cite{6951347}.

\section{Proposed Transmission Schemes and Performance Evaluation}~\label{sec:Schemes}

In this section, we describe the FTP and ATP schemes. For each scheme, we derive the density of operable D2D transmitters, BS outage probability, D2D outage probability, and average achievable D2D sum-rate. In both schemes only the operable D2D transmitters, whose available energy is $ \geq E_{th}$, will be able to access the channel.

\subsection{Scheme 1: Fixed Transmission Probability~(FTP)}

In the FTP scheme, a fixed transmission probability is equally assigned to all the operable D2D transmitters. Thus, only a subset of operable transmitters can communicate with their corresponding receivers. In order to derive the outage probability expressions and average achievable D2D sum-rate, we need to derive the probability that a D2D transmitter is operable $\pi_{o}^{\FTP}$ and the density of operable and transmitting D2D users, $\lambda_t^{\FTP}$.

According to \eqref{Eq: general energy model}, the energy model for the $i^{th}$ D2D transmitter in FTP scheme can be expressed as

 \begin{align*}
     E_n^i=E_{n-1}^i+H_{n-1}^i-{P_d} X_{n-1,o}^i  X_{n-1,t}^i,
     \numberthis
 \end{align*}

where $X_{n-1,o}^i$ and $X_{n-1,t}^i$ denote the D2D transmitter status of being operable and transmitting, respectively given as

 \begin{align*}
      X_{n-1,o}^i=
      \begin{cases}
        1 &  \hspace{1mm} {\text{if}} \hspace{4mm} E_{n-1}^i\geq E_{th}
      \\0 & \hspace{1mm}  {\text{if}} \hspace{4mm} E_{n-1}^i < E_{th},
      \end{cases}
     \numberthis
 \end{align*}

 and

 \begin{align*}
      X_{n-1,t}^i=\begin{cases} 1 & {\text{with \hspace{1mm} probability}} \hspace{3 mm} p_t^{\FTP} \\0 & {\text{with \hspace{1mm} probability}} \hspace{3 mm} 1-p_t^{\FTP},\end{cases}
     \numberthis
 \end{align*}

where $p_t^{\FTP}$ denotes the transmission probability of an operable D2D transmitter.
Let define the probability of being an operable user, i.e., $\pi_{o}^{\FTP}=\lim_{N\rightarrow \infty} \frac{1}{N}\sum^{N}_{n=0}\mathbb{E}\left[X^i_{n,o}\right]$ \cite{6609136} which is derived in Proposition \ref{prop:pftp}.

\begin{prop}\label{prop:pftp}
 The probability of being an operable D2D transmitter with FTP scheme is derived as

 \begin{align}
\pi_{o}^{\FTP}
=\begin{cases}
        \frac{\alpha \mu_1^{\frac{2}{\alpha}}}{(\alpha-2)R^2}+\frac{2\mu_1 }{(2-\alpha)R^{\alpha}} &  \hspace{1mm} {\text{if}} \hspace{4mm} \left(\frac{\eta P_b}{P_d p_t^{\FTP}}\right)^{\frac{1}{\alpha}}\leq R
      \\1 & \hspace{1mm} {\text{if}} \hspace{4mm} \left(\frac{\eta P_b}{P_d p_t^{\FTP}}\right)^{\frac{1}{\alpha}}> R,
\end{cases}
 \label{energy_ftp}
 \end{align}
 where $\mu_1=\frac{\eta P_b}{P_d p_t^{\FTP}}$.

\end{prop}

\begin{proof}
See Appendix~\ref{proof:piFTP}.
\end{proof}

As \eqref{energy_ftp} indicates, the probability of being an operable user has a reverse relation with $p_t^{\FTP}$ in the case $\left(\frac{\eta P_b}{P_d p_t^{\FTP}}\right)^{\frac{1}{\alpha}}\leq R$. However, for {the case} $\left(\frac{\eta P_b}{P_d p_t^{\FTP}}\right)^{\frac{1}{\alpha}}> R$, we have $ p_t^{\FTP}=1$. By deriving $\pi_{o}^{\FTP}$, the density of users which are operable and will transmit is $\lambda_t^{\FTP}=\lambda_d \pi_{o}^{\FTP}  p_t^{\FTP}$.

Now, using $\lambda_t^{\FTP}$, we can derive and evaluate the performance of FTP scheme.

\subsubsection{BS Outage Probability}

 The interference at BS is only caused by the transmitting operable D2D users with density $\lambda_t^{\FTP}$. The outage probability of BS is described as the probability
that the SINR at the BS is less than a predetermined threshold, i.e. $\gamma_b$. By invoking \eqref{sinr_bs_1}, the following proposition provides the BS outage probability within the FTP scheme.

\begin{prop}\label{pro: BS outage probability,FTP}
The outage probability of BS for FTP scheme can be approximated as

\begin{align}
\label{propositio: BS outage probability,FTP}
P_{out,b}^{\FTP}(\gamma_{b})\approx&1-\frac{\pi}{R K}\sum_{k=1}^{K}  a_k \sqrt{1-x_k^2} \exp\left(-\frac{\gamma_b  N_0}{P_c a_k^{-\alpha}}\right)
 \nonumber\\
 &\hspace{0em}\times\exp\!\left(\!-\pi\!\lambda_{t}^{\FTP} \Xi(\alpha) a_k^2\!P_d^\frac{2}{\alpha} P_c^\frac{-2}{\alpha}\gamma_{b}^{\frac{2}{\alpha}}\right),
\end{align}

where $\Xi(\alpha)=\Gamma\left(1-\frac{2}{\alpha}\right)\Gamma\left(1+\frac{2}{\alpha}\right)$, and $K$ 
is defined as a parameter to guarantee an accuracy-complexity
tradeoff, $x_k=cos(\frac{2k-1}{2K}\pi)$ and $a_k=\frac{R}{2}x_k+\frac{R}{2}$.
\end{prop}

\begin{proof}
See Appendix~\ref{proof:BS:FTP}.
\end{proof}

The outage expression in \eqref{propositio: BS outage probability,FTP} highlights the effect of $\pi_{o}^{\FTP}$ and $p_t^{\FTP}$ on the $P_{out,b}^{\FTP}(\gamma_{b})$. If $\pi_{o}^{\FTP}=1$, then $\lambda_t^{\FTP}=\lambda_d p_t^{\FTP}$. In this case, increasing $p_t^{\FTP}$, causes increase of $P_{out,b}^{\FTP}(\gamma_{b})$. The reason is that all the D2D transmitters will be operable and then the number of transmitters will be only dependent on $\lambda_d$ and $p_t^{\FTP}$; hence, the increase of these parameters results in a large number of transmitting users. If $\pi_{o}^{\FTP}\neq1$, then $\lambda_t^{\FTP}=\lambda_d p_t^{\FTP}\pi_{o}^{\FTP}$. In this case, $\alpha$, $\mu_1$ and $R$ specify the trends of $\lambda_t^{\FTP}$ which in turn specifies the trends of $P_{out,b}^{\FTP}(\gamma_{b})$.

\subsubsection{D2D Outage Probability}\label{prop: D2D outage probability, FTP}

According to \eqref{received signal at D2D}, the inflicted interference for a given D2D receiver is caused by the cellular user and the other transmitting operable D2D transmitters. The D2D outage probability is expressed as the probability in which the SINR at a D2D receiver is less than a predetermined threshold $\gamma_d$. The following proposition provides the exact D2D outage probability for the FTP scheme.

\begin{prop} ~\label{Prop:D2D:FTP}
With the FTP scheme, the D2D outage probability can be approximated as
\begin{align*}
            P_{out,d}^{\FTP}(\gamma_{d})&\!\approx\! 1\!-\!\exp\!\left(\!-\frac{\gamma_{d}N_{0}}{P_d{r_d^{-\alpha}}}\right)   \exp\left(-\pi \lambda_{t}^{\FTP} r_d^2 \gamma_{d}^{\frac{2}{\alpha}} \Xi(\alpha)\!\right)  \\
            & \times \frac{R\pi}{K}\sum_{k=1}^{K}\frac{\sqrt{1-x_k^2}}{1+\frac{\gamma_{d}P_cb_k^{-\alpha}}{{P_d}r_d^{-\alpha}}}f_{d_{c,i}}\left(b_k\right),
    \label{d2d_ftp_cheby}
    \numberthis
    \end{align*}
where $b_k=Rx_k+R$, and $f_{d_{c,i}}\left(r\right)$ given by
\end{prop}
\begin{align*}
        &f_{d_{c,i}}\left(r\right)=\left(\frac{2r}{R^2}\right)\left(\frac{2}{\pi}\arccos\left(\frac{r}{2R}\right)-\frac{r}{\pi R}\sqrt{1-\frac{r^2}{4R^2}}\right) ,\\&  
        \label{distance distribution}
\numberthis
\end{align*}

for $0 \leq r \leq 2R$, denotes the PDF of the distance between two randomly distributed nodes in a disk with radius $R$ \cite{moltchanov2012distance}. 
 
\begin{proof}
See Appendix~\ref{proof:D2D:FTP}.

\label{proof: D2D outage probability, FTP}
\end{proof}

The proposition \ref{Prop:D2D:FTP} explicitly reveals that the D2D outage probability determined by $p_t^{\FTP}$ and $\pi_{o}^{\FTP}$ in which these two parameters specify the $\lambda_t^{\FTP}$. By increasing $p_t^{\FTP}\pi_{o}^{\FTP}$, $\lambda_t^{\FTP}$ increases and accordingly the outage probability increases.

\subsubsection{Average achievable D2D Sum-Rate}

The average achievable D2D sum-rate with the FTP scheme is derived in the Proposition \ref{sumra_ftp} by considering the fact that the D2D users only transmit in UL and the average number of transmitting operable transmitters equals to $N_{t}^{\FTP}={\lambda_{t}^{\FTP}\pi R^2}$.

\begin{prop}~\label{Prop:Rate:FTP}
\label{sumra_ftp}
    The average achievable D2D sum-rate of the FTP scheme can be expressed as

    \begin{align*}
            R_{s}^{\FTP} &\!= \!\frac{{\lambda_{t}^{\FTP}\pi R^2}}{2 \hspace{1mm}{\text{Ln}2}}\int_{0}^{\infty}\!\!\!\frac{1-\!P_{out,d}^{\FTP}\left(x\right)}{1+x} dx. 
    \label{sumrate_ftp}
    \numberthis
    \end{align*}
    \label{D2D sum rate, FTP}

\end{prop}

\begin{proof}
See Appendix~\ref{APX:pr:D3D rate}.
\end{proof}

 It is seen that the trends of  $\lambda_t^{\FTP}$ and $P_{out,d}^{\FTP}$ which are dependent on $p_t^{\FTP}$ and $\pi_{o}^{\FTP}$, specify the $R_s^{\FTP}$. 

\subsection{Scheme 2: Adaptive Transmission Probability~(ATP)}

In this scheme, an operable D2D transmitter senses the channel before accessing it, based on which it can decide to transmit or remain idle if an ongoing transmission is detected. Note that a non-operable D2D transmitter just stays idle. An elementary sensing period has a duration of $T_s$. Each D2D transmitter selects a random number between 0 and $1$ which is denoted by $t_s^{i}$ for the $i^{th}$ D2D transmitter i.e. $t_s^{i} \sim U[0, 1]$ as in~\cite{nguyen2012stochastic}. Based on the selected number and the sensing duration $T_s$, the maximal sensing period is set to $t_s^i T_s$. As the sensing period is considered very short, we can adopt the approximation that the cellular user also starts to transmit after $T_s$, which simplifies the expressions. A subset of operable transmitters that do not detect any nearby transmission start to access to the channel, while the others remain idle. 

In this scheme, the available energy for the $i^{th}$ D2D transmitter can be expressed as

  \begin{align*}
     &E_n^{i}=\!E_{n-1}^{i}+H_{n-1}^{i}\!-{P_s}{t_s^{i}T_s} X^i_{n-1,o}\\&\hspace{1em}-{P_d}({1-t_s^{i}T_s}){X^i}_{n-1,o}X_{n-1,t}^i,
     \numberthis
 \end{align*}

where $P_s$ denotes the consumed power during sensing period. We assume that the consumed energy for sensing is negligible compared to $E_{th}$, such that it cannot affect the state change from operable to non-operable. Hence, if $E_n^{i} \geq E_{th}$ for the $i^{th}$ user, then $E_n^{i}-P_s {t_s^i}T \geq E_{th}$. $X_{n-1,o}^i$ is the state of being operable at the beginning of the UL sub-slot:

 \begin{align*}
      X_{n-1,o}^i=\begin{cases} 1 & {\text{if}} \hspace{4mm} E_{n-1}^i\geq E_{th} \\0 &  {\text{if}} \hspace{4mm} E_{n-1}^i< E_{th}, \end{cases}
      \label{Y}
     \numberthis
 \end{align*}

 and $X_{n-1,t}^i$ is the channel access capability:

  \begin{align*}
      X_{n-1,t}^i=\begin{cases} 1 & \text{{w}ith \hspace{1mm} probability} \hspace{3 mm} p_t^{\ATP} \\0 &  \text{with \hspace{1mm} probability} \hspace{3 mm} 1-p_t^{\ATP}.\end{cases}
     \numberthis
     \label{Z}
 \end{align*}

We denote the probability of being an operable D2D user with $\pi_{o}^{\ATP}=\lim_{N\rightarrow \infty} \frac{1}{N}\sum^{N}_{n=0}\mathbb{E}\left[X^i_{n,o}\right]$. By this definition, and following the same approaches as in Proposition~\ref{prop:pftp}, $\pi_{o}^{\ATP}$ can be obtained as

\begin{align*}
\pi_{o}^{\ATP}&=\!\mathbb{E}_{d_{b,i}}\!\left[\min\left(1,\frac{\mathbb{E}_{h_{b,i}}\left[{\eta P_{b}
|h_{b,i}|^{2}} (d_{b,i})^{-\alpha}\right]}{P_s \mathbb{E}[t_s^iT_s]+P_d\mathbb{E}[1\!-t_s^iT_s]p_t^{\ATP}}\right)\right]\\
&=\int_0^{R}\min\left(\!1,\!\frac{{\eta P_b} r^{-\alpha}}{P_s \frac{T_s}{2}+P_d (1-\frac{T_s}{2})p_t^{\ATP}}\!\right)\frac{2r}{R^2}dr.
\label{transmiss ATP}
\numberthis
\end{align*}

We will have the following cases for $i^{th}$ and $j^{th}$ operable D2D transmitters:
\begin{itemize}
    \item The $i^{th}$ and $j^{th}$ operable D2D transmitters will transmit if none of them detects another transmission.
    \item The $i^{th}$ operable D2D transmitter will transmit if $t_s^i<t_s^j$ or $t_s^i\geq t_s^j$ but the $i^{th}$ operable D2D transmitter does not detect another transmitted signal. In fact, this means that the $j^{th}$ operable D2D transmitter is not in the protection region of the $i^{th}$ operable D2D transmitter. More specifically, the protection region is defined as a circular region around each operable D2D transmitter with radius $\left({\frac{P_d|h_{i,j}|^2}{\beta_{th}}}\right)^{\frac{1}{\alpha}}$, where $|h_{i,j}|^2$ is the channel gain between $i^{th}$ and $j^{th}$ transmitters. The average protection region around an operable D2D transmitter can be obtained as

\begin{align*}
        r_p&= \mathbb{E}_{h_{i,j}}\left[\left({\frac{P_d|h_{i,j}|^2}{\beta_{th}}}\right)^{\frac{1}{\alpha}}\right]=\left(\frac{P_d}{\beta_{th}}\right)^{\frac{1}{\alpha}} \Gamma\left(1+\frac{1}{\alpha}\right),
\label{rp}
\numberthis
\end{align*}

{where the received power from the nearby interfering D2D transmitter should be less than the protection threshold $\beta_{th}$.}
The fact is that the D2D receivers should be protected from the interference of the other transmitters. By assuming 
$r_p >> r_d$, a D2D receiver is also protected by the sensing done by its transmitter \cite{sakr2015cognitive}. 
\end{itemize}

{The transmission probability of ATP scheme, $p_{t}^{\ATP}$, by using \cite[Proposition 13]{nguyen2012stochastic} and \cite{7073589}, can be written as 

   \begin{equation}
    \begin{aligned}
           & p_{t}^{\ATP}=\frac{1-\exp(-W \pi_{o}^{\ATP})}{W \pi_{o}^{\ATP}},
   \label{eq: ATPtransmission probability2}
   \end{aligned}
   \end{equation}

where $ W=\frac{2\pi \Gamma\left(\frac{2}{\alpha}\right)\lambda_d}{\alpha \left (\frac{\beta_{th}}{P_d}\right)^{\frac{2}{\alpha}}}$.}

Then, by invoking \eqref{transmiss ATP} and \eqref{eq: ATPtransmission probability2}, $\pi_{o}^{\ATP}$ can be derived as a function of the system parameters

\begin{align*}
&\pi_{o}^{\ATP}\!\!=\!\int_0^R\!\!\!\!\min\Bigg(1,\frac{\eta P_b r^{-\alpha}}{{P_s \frac{T_s}{2}}\!+\!{P_d (1-{\frac{T_s}{2}})\frac{1-\exp(-W\pi_{o}^{\ATP})}{W\pi_{o}^{\ATP}}}}\Bigg)\frac{2r}{R^2}dr.
\label{equation of ATP}
\numberthis
\end{align*}

Now, we turn our attention to derive the outage probabilities and the average achievable D2D sum-rate for the ATP scheme.

\subsubsection{BS Outage Probability}
The BS is in outage if the received SINR at the BS is less than $\gamma_{b}$. For ATP scheme, we have the following result.

\begin{prop}\label{prop:out:atp}
BS outage probability within ATP scheme can be obtained as

\begin{align}\label{eq:BS,ATP,cheby}
            &P_{out,b}^{\ATP}(\gamma_{b})\approx 1-\frac{\pi}{K R}\sum_{k=1}^{K} a_k \sqrt{1-{x_k}^2} \exp\left(\!-\frac{\gamma_{b}N_{0}}{P_c{a_k^{-\alpha}}}\right) \nonumber\\
            &\hspace{2em} \times \exp\!\left(\!-\pi\!\lambda_{t}^{\ATP} \!\gamma_{b}^{\frac{2}{\alpha}} a_k^2\!P_d^\frac{2}{\alpha} P_c^\frac{-2}{\alpha}
            \Xi(\alpha)\right).
    \end{align}

\end{prop}
\begin{proof}
The proof follows similar steps as those in Proposition~\ref{pro: BS outage probability,FTP}, and it is thus omitted.
\end{proof}

Proposition~\ref{prop:out:atp}
implies that $P_{out,b}^{\ATP}(\gamma_{b})$ is the same as for FTP scheme, considering the same number of D2D transmissions which, as Proposition~\ref{prop:out:atp} denotes, $P_{out,b}^{\ATP}(\gamma_{b})$ is a function of $\lambda_t^{\ATP}$. From \eqref{equation of ATP}, it follows that $p_t^{\ATP}$ and $\pi_{o}^{\ATP}$ are related to each other due to sensing; this results in an  algorithm that adapts the density of transmitting operable users.

\subsubsection{D2D Outage Probability}

D2D outage probability for the ATP scheme can be written as

    \begin{align*}
            P_{out,d}^{\ATP}\left(\gamma_{d}\right)
            &=1-\exp\left(-\frac{\gamma_{d}N_{0}}{P_d{{r_d}^{-\alpha}}}\right)\\&\times\mathbbm{E}_{\phi_{t},h_{j,i}}
            \left[\prod_{j\in\phi_{t}} \exp\left(-\frac{\gamma_{d}|h_{j,i}|^2 d_{j,i}^{-\alpha}}{{r_{d}^{-\alpha}}}\right)\right]\\&\times\int_{0}^{2R}\left(\frac{1}{1+\frac{\gamma_{d}P_c{r}^{-\alpha}}{{P_d}{r_d}^{-\alpha}}}\right) \times f_{d_{c,i}}\left(r\right) dr,
            \numberthis
            \end{align*}

where

            \begin{align*}
            &\mathbbm{E}_{\phi_{t},h_{j,i}}
            \left[\prod_{j\in\phi_{t}} \exp\left(-\frac{\gamma_{d}|h_{j,i}|^2 d_{j,i}^{-\alpha}}{{r_{d}^{-\alpha}}}\right)\right] \nonumber\\
            &=\!\exp\!\left(\!-\!2\pi\!\lambda_{t}^{\ATP}\!\int_{r_p}^{\!\infty}\!\!\left(\!\!1-\!\mathbbm{E}_{h_{\!j,i\!}}\!\left[\!\exp\!\left(\frac{-\gamma_{d}|h_{j,i}|^2 r^{-\alpha}\!}{ {r_{d}^{-\alpha}}}\!\right)\right]\right)\!rdr\!\right)\\
            & =\exp\left(-\pi\!\lambda_{t}^{\ATP} \int_{{r_p}^2}^{\infty}\frac{dv}{1+\frac{{r_{d}^{-\alpha}} v^{\frac{\alpha}{2}}}{\gamma_{d} }}\right).
            \numberthis
            \label{proofBSatp}
    \end{align*}

By setting $r^2=v$, we can express $P_{out,d}^{\ATP}(\gamma_{d})$ as

   \begin{equation}
    \begin{aligned}
           P_{out,d}^{\ATP}(\gamma_{d})
           &=1-\exp\left(-\frac{\gamma_{d}N_{0}}{P_d{{r_d}^{-\alpha}}}\right)\\&\times\exp\left(-\pi \lambda_{t}^{\ATP} \int_{r_p^2}^{\infty} \frac{dv}{1+\frac{  r_d^{-\alpha} v^{\frac{\alpha}{2}}}{\gamma_{d} }} \right)\\ &\times \underbrace{\int_{0}^{2R}\left(\frac{1}{1+\frac{\gamma_{d}P_c{r}^{-\alpha}}{{P_d}{r_d}^{-\alpha}}}\right) \times f_{d_{c,i}}\left(r\right) dr}_{O_1^{\ATP}}.
   \label{poutd2dtATP}
   \end{aligned}
   \end{equation}
where by using Gaussian-Chebyshev
quadrature, we can express ${O_1^{\ATP}}$ as
\begin{align*}
 {O_1^{\ATP}}\approx\frac{R\pi}{K}\sum_{k=1}^{K}\frac{\sqrt{1-{x_k}^2}}{1+\frac{\gamma_{d}P_c{b_k}^{-\alpha}}{{P_d}{r_d}^{-\alpha}}}f_{d_{c,i}}\left(b_k\right).
 \numberthis
\end{align*}

The difference between the D2D outage probability of FTP and ATP schemes arises from  the interference from the different set of D2D transmitters due to the sensing property. However, the terms related to the noise and the inflicted interference from the cellular user are the same for FTP and ATP schemes.

\begin{figure}[th]
\centering
\includegraphics[width=80mm, height=60mm]{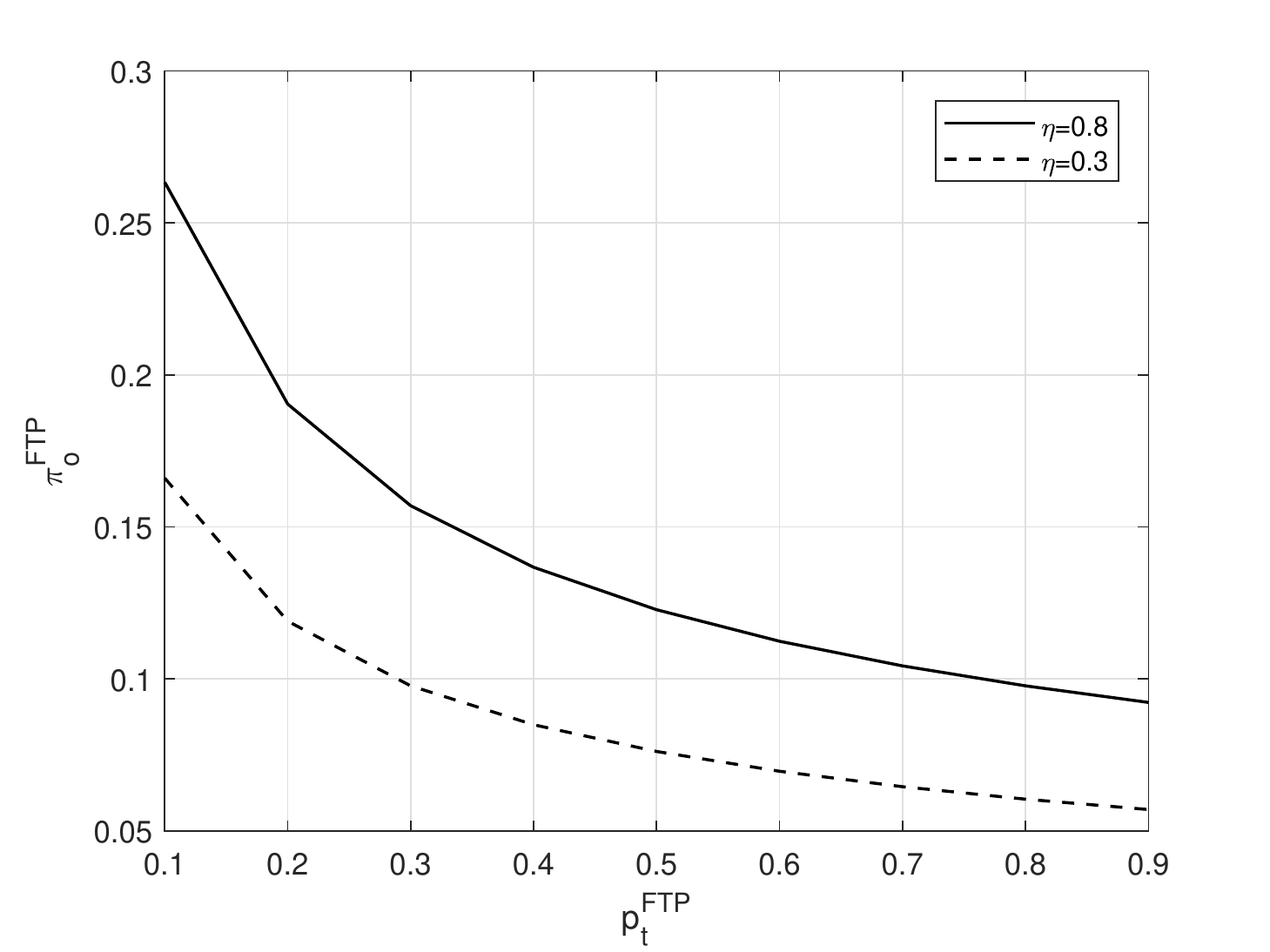}
\caption{The probability of being operable, $\pi_{o}^{\FTP}$, versus $p_t^{\FTP}$ for different $\eta$ i.e. 0.3 and 0.8 for FTP scheme.}
\label{pr_FTP_pon}
\end{figure}

\subsubsection{Average achievable D2D Sum-Rate}

In this scheme, the number of transmitters will be $N_{t}^{\ATP}=\lambda_{t}^{\ATP}\pi R^2$. Then, the average achievable D2D sum-rate with regards to that the D2D transmitter can only transmit on UL sub-slot and based on the derived D2D outage probability is given by

\begin{equation}
\begin{aligned}
         R_{s}^{\ATP}&=\mathbbm{E}\left[\sum_{i=1}^{N_{t}^{\ATP}}\frac{1-t_s^{i}T_s}{2}\log_{2}(1+\Gamma_{i})\right]
         \\&=\frac{1-{\frac{T_s}{2}}}{2} \!\frac{{\lambda_{t}^{\ATP}\!\pi\!R^2}}{\text{Ln}2}\!\int_{0}^{\infty}\!\frac{1-P_{out,d}^{\ATP}\left(x\right)}{1+x} \!dx,
\label{sumrate_atp}
\end{aligned}
\end{equation}

where in the last step, we use the fact that  $\mathbbm{E}[T_st_s^{i}]=\frac{T_s}{2}$ and $p_{out,d}^{\ATP}$ is invoked from \eqref{poutd2dtATP}.

\begin{rem}
The expressions in \eqref{sumrate_ftp} and \eqref{sumrate_atp} are not simple
enough to provide immediate insight, but they are general and
fast to evaluate using popular scientific software packages such as Matlab and Mathematica.
\end{rem}

\section{Numerical Results and Discussions}~\label{sec:num}

In this section, we evaluate the performance of FTP and ATP schemes for EH based D2D underlying cellular network. For both schemes, we present the results for D2D and cellular outage probabilities and also average achievable D2D sum-rate. Then, we discuss the effect of important parameters such as $\lambda_d$, $\eta$ $p_t^{\FTP}$, and $\beta_{th}$ on these performance metrics.

We set the transmit power of BS and cellular transmitters to $P_{b}=44$~dBm and $P_{c}=10$~dBm. The D2D transmission, D2D sensing, and noise power are set to $P_{d}=-10$~dBm, $P_s=-30$~dBm and $N_0=-90$~dBm, respectively. The D2D and cellular transmitters are scattered in a circular region with radius $R=100$~m. In this region, each D2D receiver is placed in the distance $r_d=5$~m from its transmitter. The unit variance Rayleigh fading is considered for all links. Moreover, the path loss exponent for the cellular and D2D is set to $\alpha=4$.

The effect of transmission probability on the probability of being an operable user is
illustrated in Fig. \ref{pr_FTP_pon}.
Increasing the transmission probability is followed by consuming more energy which diminishes the energy level of the battery and the probability of being an operable user which is depicted in this figure. We compare this probability for $\eta=0.3$ and $\eta=0.8$. As expected, higher $\eta$ provides a higher amount of harvested energy which increases the probability of being an operable user.
$p_t^{\FTP}$ and $\pi_{o}^{\FTP}$ are independent of the density of D2D transmitters i.e. $\lambda_d$, then these probabilities are equal for all $\lambda_d$.
It is worthy to mention that as Fig. \ref{pr_FTP_pon} shows, in the case of $\eta=0.3$ and $p_t^{\FTP}=0.9$, the probability of being operable goes to 0 which means none of the D2D transmitters can transmit. This indicates the importance of EH designing and transmission probability $p_t^{\FTP}$ parameters on the EH-based D2D networks.

Fig. \ref{trns_atp} depicts the transmission probability and the probability of being an operable user as a function of protection power threshold i.e. $\beta_{th}$ for the ATP scheme.
As $\beta_{th}$ increases, the transmission probability increases, while the probability of being an operable user decreases. The reason is that by increasing $\beta_{th}$, more users will have the opportunity to transmit and then the transmission probability of D2D transmitters increases. 
In this regard, similar to the FTP scheme, the probability of being operable decreases by increasing the transmission probability. 
Transmission probability trends are also examined by increasing the density of D2D transmitters, $\lambda_d$. It can be observed that by increasing $\lambda_d$, the transmission probability is reduced, while the probability of being  operable is increased. 
This is due to the fact that in the ATP scheme, the sensing capability enables an adaptive transmission probability based on the density of D2D transmitters.

\begin{figure}[t!]
\centering
\includegraphics[width=80mm, height=60mm]{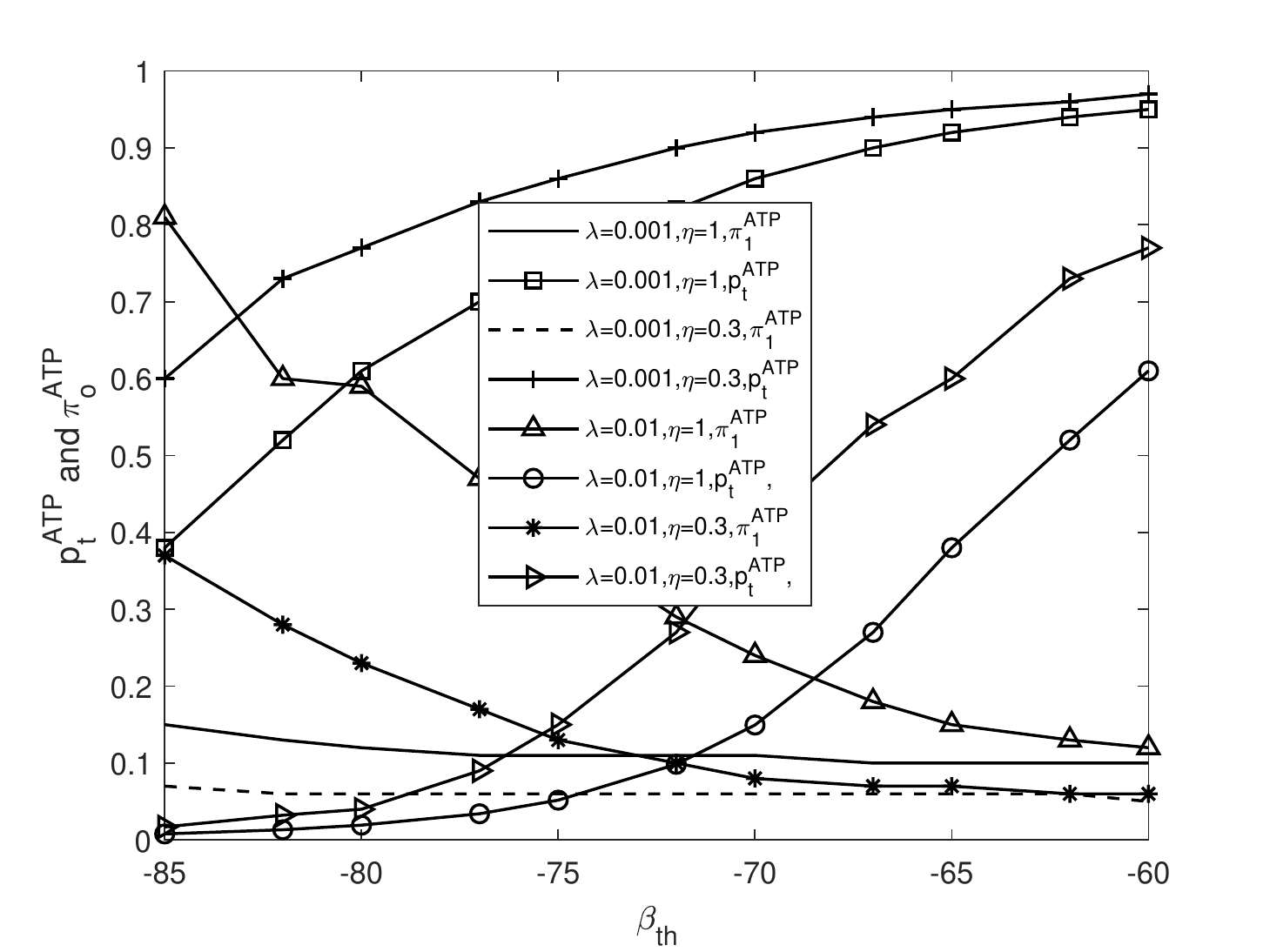}

\caption{The probability of being operable, $\pi_{o}^{\ATP}$, and the transmission probability, $p_t^{\ATP}$, versus $\beta_{th}$ and for $\lambda_d \in \{0.01,0.001\}$ and $\eta \in \{0.3,0.8\}$}
\label{trns_atp}
\end{figure}

\begin{figure}[t!]
\centering
\includegraphics[width=80mm, height=60mm]{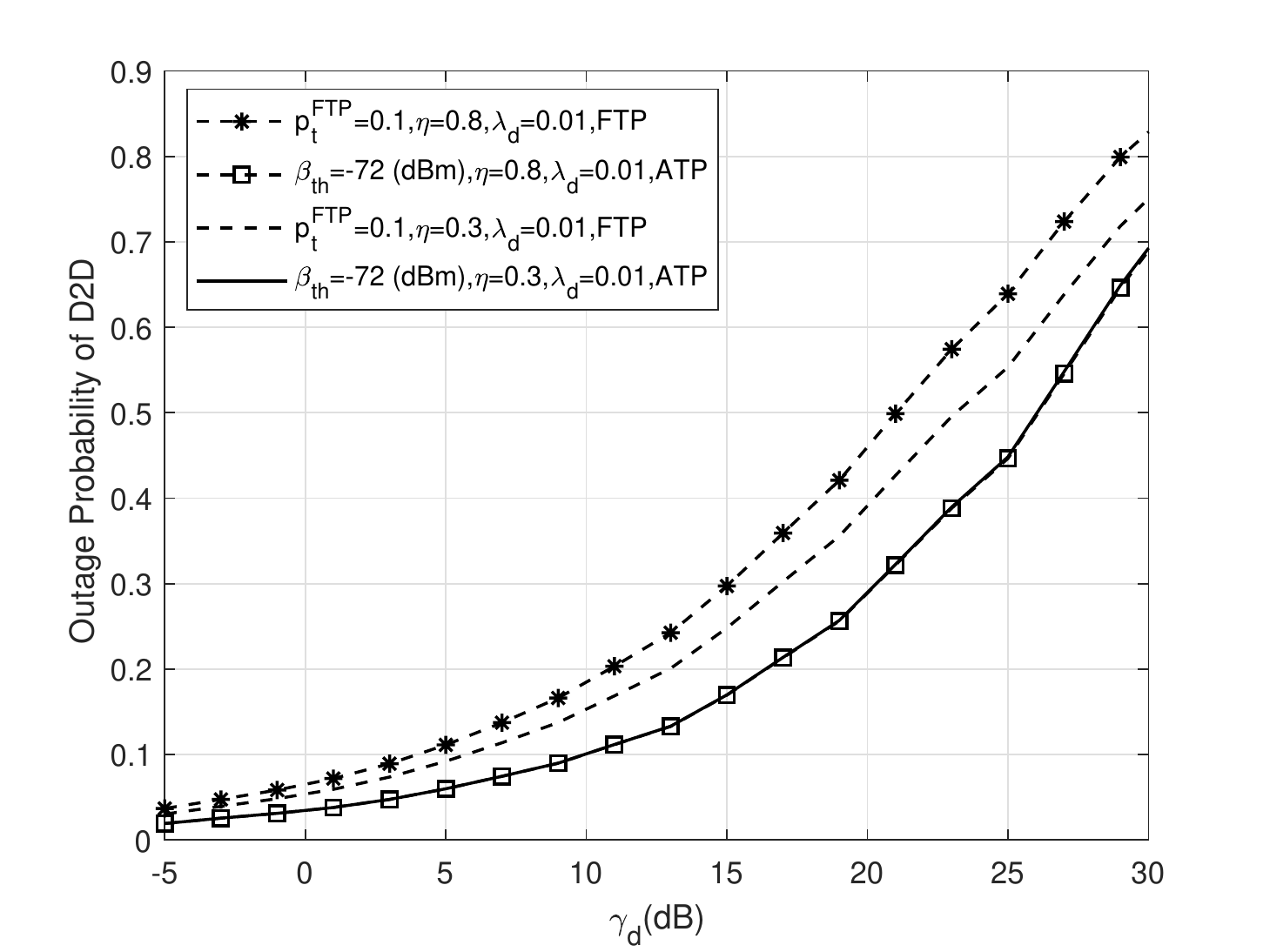}
\caption{D2D outage probability versus $\gamma_{d}$ for FTP and ATP schemes and for $\beta_{th}=-72$ dBm and $p_t^{\FTP}=0.1$. }
\label{d2d_ftp_atp_2020}
\end{figure}

\begin{figure}
\centering
\includegraphics[width=80mm, height=60mm]{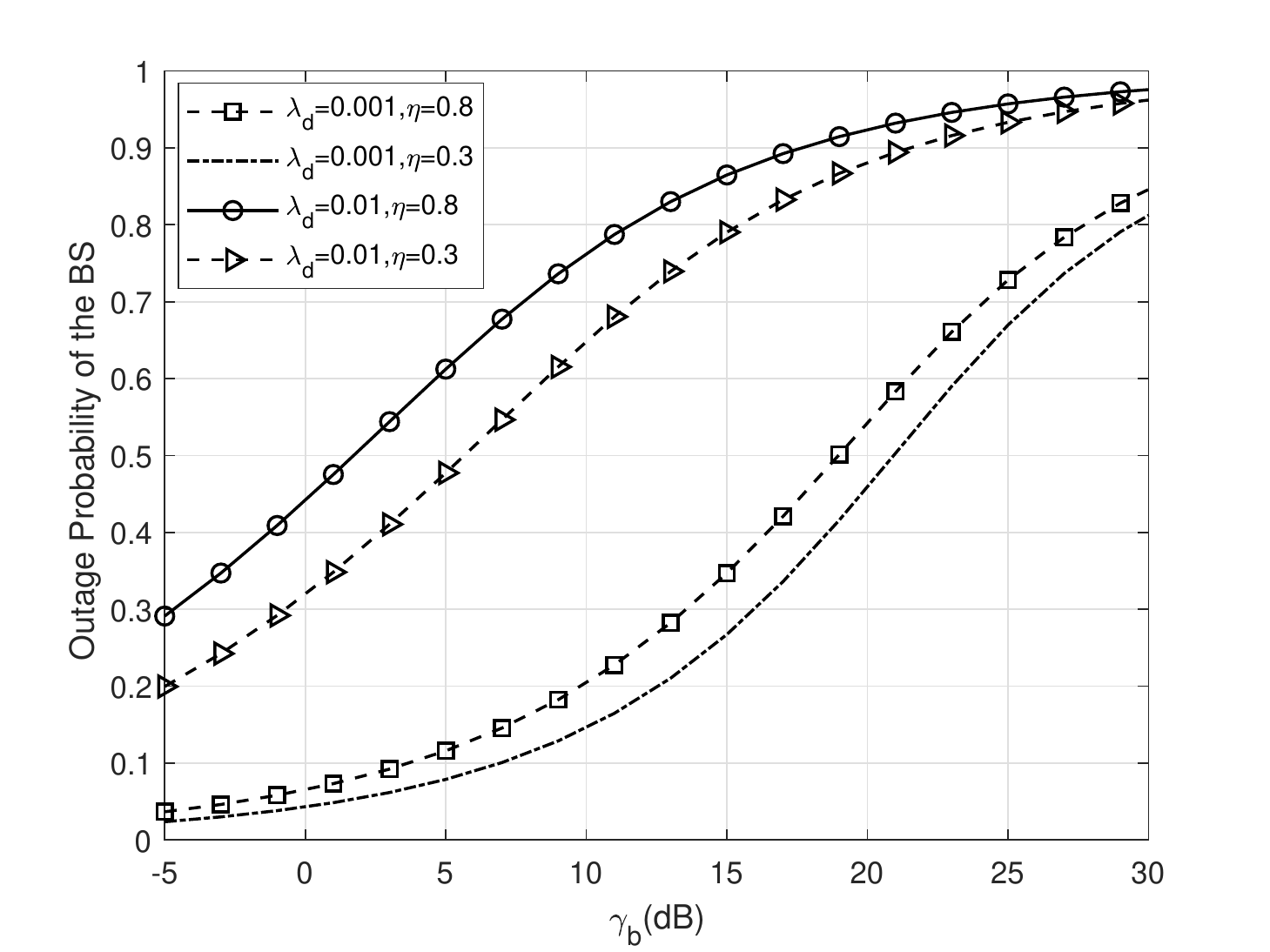}

\caption{BS outage probability versus $\gamma_{b}$ with $\lambda_t^{\FTP}=\lambda_t^{\ATP}$, for  $\lambda_d \in \{0.01,0.001\}$ and $\eta \in \{0.3,0.8\}$ and $p_t^{FTP}=0.1$.}
\label{bs_ftp_atp_2020}
\end{figure}

 Fig. \ref{d2d_ftp_atp_2020} shows the outage probability for D2D transmitters with FTP and ATP schemes as a function of $\gamma_d$. In order to allow a fair comparison, we set $\beta_{th}=$ -72~dBm and $p_t^{\FTP}=0.1$ for $\eta=0.8$ and $\lambda_d=0.01$ which provide the identical $\lambda_t^{\FTP}$ and $\lambda_t^{\ATP}$. It can be observed that ATP scheme achieves lower outage probability for D2D transmitters, especially for the higher $\gamma_d$. This result highlights the importance of channel sensing which leads to select the D2D transmitters located far enough to each other, which result in efficient interference management framework for the considered D2D underlaid cellular system. In this figure, we also include the results for $\eta=0.3$ with the same parameters. This provides the equal $\lambda_t^{\ATP}$ due to the adaptive property of the ATP scheme and also the same $\beta_{th}$ which translates to the same protection region for each D2D user. In this case, as expected, lower outage probability for FTP scheme is achieved compared to the case $\eta=0.8$.

We show the BS outage probability as a function of $\gamma_b$ in Fig. \ref{bs_ftp_atp_2020}. The BS outage probability of FTP and ATP schemes with the equal density of transmitting operable users i.e. $\lambda_t^{\FTP}=\lambda_t^{\ATP}$ are identical. Higher $\eta$ and $\lambda_d$ provide higher BS outage probability due to the higher number of transmitting operable users. It can be also seen that the BS outage probability is a monotone increasing function of $\gamma_b$. This figure also reveals that, for higher $\gamma_b$ and D2D densities larger than $\lambda_d=0.01$, the BS outage probability tends to 1, which means that none of the transmitted signals by the cellular users can be received correctly. In this regard, to provide better performance for cellular communication with the aforementioned parameters, we need to reduce the transmitting operable D2D users that send on the same cellular channel. 

\begin{figure}
\centering

\includegraphics[width=80mm, height=60mm]{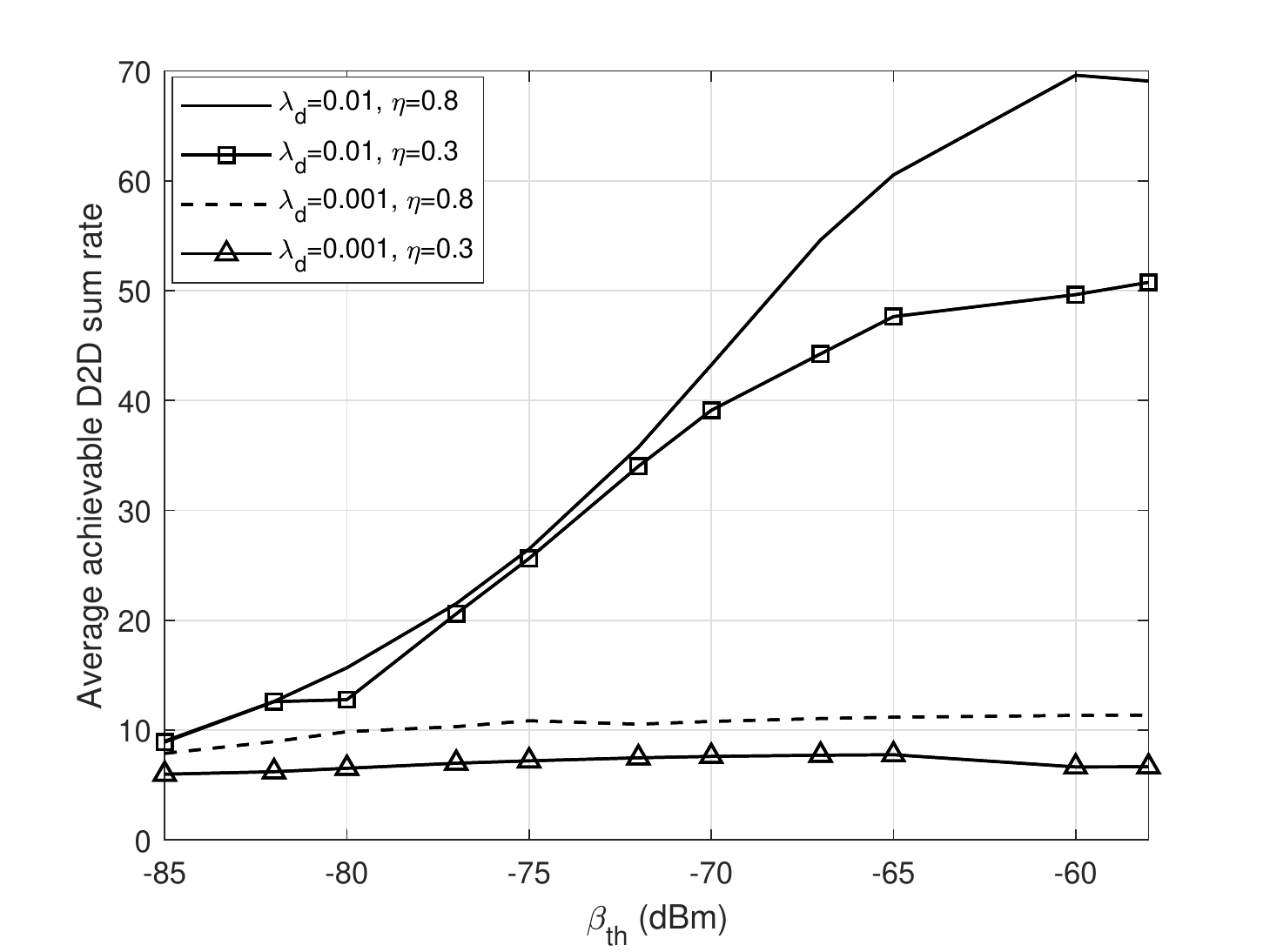}
    \caption{Average achievable D2D sum-rate versus $\beta_{th}$ for ATP scheme.}
\label{rate_atp}
\end{figure} 

In Fig. \ref{rate_atp}, we evaluate the potential tradeoff between D2D outage probability and the number of transmitting operable users. In this regard, we show the average achievable D2D sum-rate for ATP scheme as a function of $\beta_{th}$. Increasing $\beta_{th}$ results in a reduction of the protection radius which also translates to an increase in the density $\lambda_t^{\ATP}$ of transmitting operable D2Ds. It can then be seen that increasing $\beta_{th}$ has two effects on the system performance, which define the trends of D2D sum rate. The first effect is the increase of the number of transmitting operable D2Ds. The second effect is an increased D2D outage probability. Thus, selecting the best $\beta_{th}$ plays an important role in the value of D2D sum rate. As an example, for $\lambda_d=0.01$ and $\eta=0.8$, we can observe the increase of sum-rate with $\beta_{th}$ from -85 dBm to -60 dBm, but after that the sum-rate decreases. For higher $\lambda_d$ and $\eta$, the D2D sum-rate increases, which is more pronounced at higher $\beta_{th}$ due to the increased number of transmitting operable users. 

\section{Conclusion}~\label{sec:conc}

This paper has presented an interference management mechanism with a TDD protocol for a RFEH-based D2D network. The D2D transmitters follow a time-division protocol and they are allowed to transmit during the UL sub-slots, while they harvest energy during the DL sub-slots. With this protocol, a D2D transmitter can harvest sufficient energy to become operable and transmit during the UL sub-slot.

In this regard, we have presented an accurate energy model for analysis of the available energy of D2D transmitters. In order to control the interference created by the operable D2D transmitters, we have proposed two schemes, with fixed and adaptive transmission probabilities, denoted as FTP and ATP, respectively. In FTP, the operable D2D transmitters initiate a transmission with a fixed probability. In ATP, an operable D2D transmitter senses the channel and stays idle if the channel is busy. 
The evaluation results highlight the importance of the EH design parameters as well as the effects of transmission probability on the system performance.
Moreover, the results reveal that channel sensing significantly affects the outage probability and D2D sum-rate of the system. An interesting step for future work would be consideration of multiple cellular users and channels and also the possibility for channel selection for a D2D transmission.
\vspace{-1.8mm}
\begin{appendices}
\vspace{-1mm}
\section{proof of proposition \ref{prop:pftp}}\label{proof:piFTP}

The probability of D2D transmitter being an operable, can be expressed as~\cite{6609136}
\vspace{-1.5mm}
\begin{align}\label{eq:piftp:1}
\pi_{o}^{\FTP}=\lim_{N\rightarrow \infty} \frac{1}{N}\sum_{n=0}^{N}\mathbb{E}\left[X_{n,o}^i\right].
\end{align}

 In order to derive~\eqref{eq:piftp:1}, and by considering the fixed distance between D2D transmitters and BS, we investigate two different cases: $\eta P_b d_{b,i}^{-\alpha}<{P_d p_t^{\FTP}}$ and $\eta P_b d_{b,i}^{-\alpha}>{P_d p_t^{\FTP}}$ which
are respectively corresponds to the cases where the expected of harvested energy for a D2D user in the distance $d_{b,i}$ is lower and higher than the expected of consumed energy.
Under the first condition, we can readily obtain $\pi_{o}^{\FTP}=\frac{\eta P_b d_{b,i}^{-\alpha}}{P_d p_t^{\FTP}}$. On the other hand, when $\eta P_b d_{b,i}^{-\alpha}>{P_d p_t^{\FTP}}$ by using the weak law of large numbers, $\pi_{o}^{\FTP}$ is equal to 1. Therefore, for fixed distance assumption between D2D transmitter and BS,~\eqref{eq:piftp:1} can be obtained as

\begin{align*}
\pi_{o}^{\FTP} =
\min\left(1,\frac{\eta P_b d_{b,i}^{-\alpha}}{P_d p_t^{\FTP}}\right).
\numberthis
\end{align*}

Noticing that D2D transmitter are randomly located in the distance 0 to $R$ with uniform distribution, i.e., $f_{d_{b,i}}(r)=\frac{2r}{R^2}$ and by taking the average over the location of the D2D transmitter, $\pi_{o}^{\FTP}$ is given by

\begin{align*}
\pi_{o}^{\FTP}
=\int_0^R\min\left(1,\frac{\eta P_b r^{-\alpha}}{P_d p_t^{\FTP}}\right) \frac{2r}{R^2}dr,
 \numberthis
 \end{align*}

which after some algebraic manipulation the desired result in~\eqref{energy_ftp} can be obtained.

\section{proof of proposition \ref{pro: BS outage probability,FTP}}\label{proof:BS:FTP}

The BS is in outage if the received SINR at the BS is less than $\gamma_{b}$. Therefore, by invoking~\eqref{sinr_bs_1}, $P_{out,b}^{\FTP}\left(\gamma_{b}\right)$ can be written as

\begin{align*}
        P_{out,b}^{\FTP}\left(\gamma_{b}\right)
        =1-\Pr\left(\frac{P_{c}{|h_{c,b}|^2}d_{c,b}^{-\alpha}}{\sum_{j\in \phi_{t}}  P_d{|{h_{{{j},b}}}|^2}d_{{{j},b}}^{-\alpha}+N_0} \geq \gamma_{b}\right).
\numberthis
\end{align*}

By using the fact that ${|h_{c,b}|^2}$ follows unit mean exponential distribution, $\Pr(\Gamma_{b} \geq \gamma_{b})$ can be expressed as

\begin{align*}\label{bsssoutaghe}
        &\Pr(\Gamma_{b} \geq \gamma_{b})=
        {\mathbbm{E}_{d_{c,b}}} \Bigg[\exp \Big(-\frac{\gamma_
            {b} N_0}{P_c d_{c,b}^{-\alpha}} \Big) 
            \\\times
& {\mathbbm{E}_{\phi_{t},h_{{j},b}}}\Big[\prod_{j\in {\phi_{t}}}\exp\Big(-\frac{{P_d|{h_{{j},b}|^2}
            d_{{j},b}^{-\alpha}}\gamma_{b}}{P_c d_{c,b}^{-\alpha}}\Big)\Big]\Bigg].
            \numberthis
\end{align*}

To this end, by invoking the PGFL of homogeneous PPP distribution\footnote{Let $\nu(x):\mathbb{R}^2\rightarrow[0,1]$ and $\int_{\mathbb{R}^2}{\vert1-\nu(x)\vert dx}<\infty$. When $\Phi$ is Poisson of intensity $\lambda$, the conditional generating functional is $\mathbb{E}\{\prod_{x\in\Phi}\nu(x)\}=\exp\left(-\lambda\int_{\mathbb{R}^2}[1-\nu(x)]dx\right)$ \cite{chiu1996stochastic}.}, we obtain

\begin{equation}
\begin{aligned}
        &\mathbbm{E}_{{\phi_{t}},h_{{j},b}}\!\left[\prod_{j\in {\phi_{t}}}\exp\left(\frac{{-P_d|{h_{{{j},b}}|^2}d_{{{j},b}}^{-\alpha}} \gamma_{b}}{P_{c}d_{c,b}^{-\alpha}}\right)\right]
        \\& { \!=\exp\left(-2\pi\!\lambda_{t}^{\FTP} \mathbbm{E}_{h_{{{j},b}}}\!\left[\int_0^\infty\! \left(\!\!1\!-\!\exp\!\left(\!\frac{-P_d{{|h_{{{j},b\!}}|^2}} \gamma_{b}}{P_{c}d_{c,b}^{-\alpha}{r^{\alpha}}}\!\right)\right)rdr\right]\right)}
        \\ &\!=\exp\left(\!-\pi\!\lambda_{t}^{\FTP} \!\mathbbm{E}_{h_{\!{{j},b}}}\!\!\left[
        \int_0^\infty\!\!\left(\!\!1\!-
       \!\exp\Bigg(\!\frac{-P_d{|h_{{{j},b}}|^2}\gamma_{b}}{P_{c}d_{c, b}^{-\alpha}z}\!\Bigg)\!\right)\!\!\frac{2 z^{\!\frac{2}{\alpha}-1}\!dz}{\alpha}\!\right]\right),
\label{bsoutageeq3}
\end{aligned}
\end{equation}

where in the last equality, we have used the variable change $r^\alpha = z$. By using the product rule 
\begin{equation}
\begin{aligned}
        {{\big(f(x)g(x)\big)}^{'}}= f'(x)g(x)+ f(x)g'(x),
\label{productrule}
\end{aligned}
\end{equation}

where $f'(x)$ and $g'(x)$ are the derivative of $f(x)$ and $g(x)$, we get

\begin{align*}
        &\int_0^\infty\!\left(\!1-\!\exp\!\left(\frac{-P_d {|h_{{{j},b}}|^2}\gamma_{b}\!}{\!P_{c} d_{c,b}^{-\alpha}z\!}\right)\!\right)\!\frac{2z^{\!\frac{2}{\alpha}\!-1\!}}{\alpha}  dz\\&=\int_0^{\infty}\frac{P_d {|h_{{{j},b}}|^2}\gamma_{b}}{P_{c}d^{\!-\alpha}_{c,b}}\exp\!\left(\frac{-P_d {|h_{{{j},b}}|^2}\gamma_{b}v\!}{P_{c}d_{c,b}^{-\alpha}\!}\right)v^{-\frac{2}{\alpha}}dv\\&=\left(\frac{P_d{|h_{{{j},b}}|^2}
        \gamma_{b}\!}{\!P_{c}\!\!}\right)^{\frac{2}{\alpha}}d_{c,b}^2\Gamma\Big(1-\frac{2}{\alpha}\Big),
\numberthis
\label{bsoutageeq24}
\end{align*}

where we have used the variable change $z=\frac{1}{v}$. Moreover, we have      

\begin{align*}
        &\mathbbm{E}_{{\phi_{t}},h_{{j},b}}\!\left[\prod_{j\in {\phi_{t}}}\exp\left(-\frac{{P_d|{h_{{{j},b}}|^2}d_{{{j},b}}^{-\alpha}} \gamma_{b}}{P_{c}d_{c,b}^{-\alpha}}\right)\right]\\
        &=\exp\left(\!-\pi\!\lambda_{t}^{\FTP} \Xi(\alpha) \gamma_{b}^{\frac{2}{\alpha}}\!{d_{c,b}} ^2\!P_d\!^\frac{2}{\alpha} \!P_c\!^\frac{-2}{\alpha}\right).
\numberthis
\label{bsoutageeq25}
\end{align*}

By using~\eqref{bsoutageeq25}, $P_{out,b}^{\FTP}\left(\gamma_{b}\right)$ can be expressed as

\begin{align*}        &P_{out,b}^{\FTP}\left(\gamma_{b}\right)=1-\int_0^R \exp\left(-\frac{\gamma_{b} N_0}{{P_{c}r^{-\alpha}}}\right) \\ &\times\exp\left(\!-\pi\!\lambda_{t}^{\FTP} \Xi(\alpha) \gamma_{b}^{\frac{2}{\alpha}}\!P_d\!^\frac{2}{\alpha} \!P_c\!^\frac{-2}{\alpha}r^2\right) \frac{2r}{R^2}dr.
\numberthis
\end{align*}

To this end, by applying the Gaussian-Chebyshev
quadrature method \cite{hildebrand1987introduction}, the desired result in~\eqref{propositio: BS outage probability,FTP} is obtained. 
\vspace{-2mm}
\section{proof of proposition \ref{Prop:D2D:FTP}\label{proof:D2D:FTP}}
\vspace{-1.5mm}
To derive $P_{out,d}^{\FTP}$, we have

\begin{align*}
    &P_{out,d}^{\FTP}=1-\Pr\left(\Gamma_i \geq \gamma_d\right)\\&=1\!-\!\exp\!\left(\!-\frac{\gamma_{d}N_{0}}{P_d{{r_d}^{-\alpha}}}\right)   \underbrace{\exp\left(-\pi \lambda_{t}^{\FTP} {r_d}^2 \gamma_{d}^{\frac{2}{\alpha}} \Xi(\alpha)\!\right)}_{O_1^{\FTP}}  \\
            & \times \underbrace{\int_0^{2R}\frac{1}{1+\frac{P_c \gamma_d r^{-\alpha}}{{P_d}{r_d}^{-\alpha}}}f_{d_{c,i}}(r)dr}_{O_2^{\FTP}}, 
            \numberthis    
            \label{cheby_d2d_ftp}
\end{align*}
where we have used similar steps as those in the proof of Proposition \ref{pro: BS outage probability,FTP} to derive ${O_1^{\FTP}}$.

Moreover, by applying the Gaussian-Chebyshev
quadrature method, ${O_2^{\FTP}}$ can be derived as
\begin{align*}
 {O_2^{\FTP}}\approx\frac{R\pi}{K}\sum_{k=1}^{K}\frac{\sqrt{1-{x_k}^2}}{1+\frac{\gamma_{d}P_c{b_k}^{-\alpha}}{{P_d}{r_d}^{-\alpha}}}f_{d_{c,i}}\left(b_k\right).
 \numberthis
 \label{cheby_d2d_cell_ftp}
\end{align*}
 
To this end, by substituting \eqref{cheby_d2d_cell_ftp} into the last term of \eqref{cheby_d2d_ftp}, the desired result is obtained.

\section{proof of proposition \ref{Prop:Rate:FTP}}\label{APX:pr:D3D rate}

The achievable D2D sum-rate can be written as

\begin{align*}
        R_{s}^{\FTP} &= \frac{1}{2}\mathbbm{E}\left[\sum_{n=1}^{N_{t}^{\FTP}}\log_2(1+\Gamma_n)\right]\\&=\frac{1}{2}{\lambda_{t}^{\FTP}\pi  R^2}  \mathbbm{E} \left[\log_2\left(1+\Gamma_d\right)\right] \\ &=\frac{1}{2}{\lambda_{t}^{\FTP}\pi R^2} \int_{0}^{\infty} \log_2\left(1+x\right) f_{\Gamma_d}\left(x\right)dx,
        \label{sum rate_proof_ftp}
\numberthis
\end{align*}

where $\mathbbm{E}\left[\log_2\left(1+\Gamma_d\right)\right]$ is the average achievable rate. By applying the product rule, We can write

\begin{align*}
        &\int_{0}^{\infty} \log_2\left(1+x\right) f_{\Gamma_d}\left(x\right)dx \\&=-\int_{0}^{\infty} \log_2\left(1+x\right) \left(\Pr\left({\Gamma_d}\geq x\right)\right)'dx\\&=\frac{1}{\text{Ln}2}\int_{0}^{\infty} \frac{1}{1+x} \left(1-P_{out,d}^{\FTP}\left(x\right)\right)dx.
\numberthis
\end{align*}

To this end, invoking the outage probability expression in~\eqref{d2d_ftp_cheby} yields the desired result.
\end{appendices}

\balance
\bibliographystyle{IEEEtran}
\bibliography{references}

\begin{thebibliography}{10}
\providecommand{\url}[1]{#1}
\csname url@samestyle\endcsname
\providecommand{\newblock}{\relax}
\providecommand{\bibinfo}[2]{#2}
\providecommand{\BIBentrySTDinterwordspacing}{\spaceskip=0pt\relax}
\providecommand{\BIBentryALTinterwordstretchfactor}{4}
\providecommand{\BIBentryALTinterwordspacing}{\spaceskip=\fontdimen2\font plus
\BIBentryALTinterwordstretchfactor\fontdimen3\font minus
  \fontdimen4\font\relax}
\providecommand{\BIBforeignlanguage}[2]{{%
\expandafter\ifx\csname l@#1\endcsname\relax
\typeout{** WARNING: IEEEtran.bst: No hyphenation pattern has been}%
\typeout{** loaded for the language `#1'. Using the pattern for}%
\typeout{** the default language instead.}%
\else
\language=\csname l@#1\endcsname
\fi
#2}}
\providecommand{\BIBdecl}{\relax}
\BIBdecl

\bibitem{6807945}
X.~Lin, J.~G. Andrews, A.~Ghosh, and R.~Ratasuk, ``An overview of {3GPP}
  device-to-device proximity services,'' \emph{IEEE Commun. Mag.}, vol.~52,
  no.~4, pp. 40--48, Apr. 2014.

\bibitem{6163598}
G.~Fodor, E.~Dahlman, G.~Mildh, S.~Parkvall, N.~Reider, G.~Miklós, and
  Z.~Turányi, ``Design aspects of network assisted device-to-device
  communications,'' \emph{IEEE Commun. Mag.}, vol.~50, no.~3, pp. 170--177,
  Mar. 2012.

\bibitem{8340813}
F.~{Jameel}, Z.~{Hamid}, F.~{Jabeen}, S.~{Zeadally}, and M.~A. {Javed}, ``A
  survey of device-to-device communications: Research issues and challenges,''
  \emph{IEEE Commun. Surv. Tut.}, vol.~20, no.~3, pp. 2133--2168, 3rd Quart.
  2018.

\bibitem{5910123}
C.~Yu, K.~Doppler, C.~B. Ribeiro, and O.~Tirkkonen, ``Resource sharing
  optimization for device-to-device communication underlaying cellular
  networks,'' \emph{IEEE Trans. wireless Commun.}, vol.~10, no.~8, pp.
  2752--2763, Aug. 2011.

\bibitem{6047553}
H.~Min, J.~Lee, S.~Park, and D.~Hong, ``Capacity enhancement using an
  interference limited area for device-to-device uplink underlaying cellular
  networks,'' \emph{IEEE Trans. wireless Commun.}, vol.~10, no.~12, pp.
  3995--4000, Dec. 2011.

\bibitem{6736746}
F.~{Boccardi}, R.~W. {Heath}, A.~{Lozano}, T.~L. {Marzetta}, and P.~{Popovski},
  ``Five disruptive technology directions for {5G},'' \emph{IEEE Commun. Mag.},
  vol.~52, no.~2, pp. 74--80, Feb. 2014.

\bibitem{8643973}
P.~S. {Bithas}, K.~{Maliatsos}, and F.~{Foukalas}, ``An {SINR}-aware joint mode
  selection, scheduling, and resource allocation scheme for {D2D}
  communications,'' \emph{IEEE Trans. Veh. Tech.}, vol.~68, no.~5, pp.
  4949--4963, May 2019.

\bibitem{8766143}
Z.~{Zhang}, Y.~{Xiao}, Z.~{Ma}, M.~{Xiao}, Z.~{Ding}, X.~{Lei}, G.~K.
  {Karagiannidis}, and P.~{Fan}, ``{6G} wireless networks: Vision,
  requirements, architecture, and key technologies,'' \emph{IEEE Veh. Technol.
  Mag.}, vol.~14, no.~3, pp. 28--41, Sep. 2019.

\bibitem{8869795}
K.~{Lee}, J.~{Hong}, H.~{Seo}, and W.~{Choi}, ``Learning-based resource
  management in device-to-device communications with energy harvesting
  requirements,'' \emph{IEEE Trans. Commun.}, vol.~60, no.~1, pp. 402--413,
  Jan. 2020.

\bibitem{sakr2015cognitive}
A.~H. Sakr and E.~Hossain, ``Cognitive and energy harvesting-based {D2D}
  communication in cellular networks: Stochastic geometry modeling and
  analysis,'' \emph{IEEE Trans. Commun.}, vol.~63, no.~5, pp. 1867--1880, May
  2015.

\bibitem{6951347}
X.~{Lu}, P.~{Wang}, D.~{Niyato}, D.~I. {Kim}, and Z.~{Han}, ``Wireless networks
  with {RF} energy harvesting: A contemporary survey,'' \emph{IEEE Commun.
  Surv. Tut.}, vol.~17, no.~2, pp. 757--789, 2nd Quart. 2015.

\bibitem{8664000}
X.~{Liu} and N.~{Ansari}, ``Toward green {IoT}: Energy solutions and key
  challenges,'' \emph{IEEE Commun. Mag.}, vol.~57, no.~3, pp. 104--110, Mar.
  2019.

\bibitem{ulukus2015energy}
S.~Ulukus, A.~Yener, E.~Erkip, O.~Simeone, M.~Zorzi, P.~Grover, and K.~Huang,
  ``Energy harvesting wireless communications: A review of recent advances,''
  \emph{IEEE J. Sel. Areas Commun.}, vol.~33, no.~3, pp. 360--381, Mar. 2015.

\bibitem{6575083}
S.~Lee, R.~Zhang, and K.~Huang, ``Opportunistic wireless energy harvesting in
  cognitive radio networks,'' \emph{IEEE Trans. wireless Commun.}, vol.~12,
  no.~9, pp. 4788--4799, Sep. 2013.

\bibitem{6042301}
J.~G. {Andrews}, F.~{Baccelli}, and R.~K. {Ganti}, ``A tractable approach to
  coverage and rate in cellular networks,'' \emph{IEEE Trans. Commun.},
  vol.~59, no.~11, pp. 3122--3134, Nov. 2011.

\bibitem{kusaladharma2017performance}
S.~Kusaladharma and C.~Tellambura, ``Performance characterization of spatially
  random energy harvesting underlay {D2D} networks with transmit power
  control,'' \emph{IEEE Trans. Green Commun. and Netw.}, vol.~2, no.~1, pp.
  87--99, Mar. 2018.

\bibitem{8069034}
L.~Shi, L.~Zhao, K.~Liang, and H.~Chen, ``Wireless energy transfer enabled
  {D2D} in underlaying cellular networks,'' \emph{IEEE Trans. Veh. Technol.},
  vol.~67, no.~2, pp. 1845--1849, Feb. 2018.

\bibitem{7782752}
R.~Atat, L.~Liu, N.~Mastronarde, and Y.~Yi, ``Energy harvesting-based
  {D2D}-assisted machine-type communications,'' \emph{IEEE Trans. Commun.},
  vol.~65, no.~3, pp. 1289--1302, Mar. 2017.

\bibitem{yang2016heterogeneous}
H.~H. Yang, J.~Lee, and T.~Q. Quek, ``Heterogeneous cellular network with
  energy harvesting-based {D2D} communication,'' \emph{IEEE Trans. wireless
  commun.}, vol.~15, no.~2, pp. 1406--1419, Oct. 2016.

\bibitem{7752450}
M.~Xie, X.~Jia, M.~Zhou, and L.~Yang, ``Study on energy efficiency of {D2D}
  underlay massive {MIMO} networks with power beacons,'' in \emph{Proc. Intl.
  Conf. wireless Commun. Signal Process. (WCSP)}, Oct. 2016, pp. 1--5.

\bibitem{8353158}
H.~Wang, J.~Wang, G.~Ding, and Z.~Han, ``{D2D} communications underlaying
  wireless powered communication networks,'' \emph{IEEE Trans. Veh. Technol.},
  vol.~67, no.~8, pp. 7872--7876, Aug. 2018.

\bibitem{ding2016dynamic}
J.~Ding, L.~Jiang, and C.~He, ``Dynamic spectrum allocation for energy
  harvesting-based underlaying {D2D} communication,'' in \emph{Proc. IEEE Veh.
  Technol. Conf. (VTC)}, May 2016, pp. 1--5.

\bibitem{8536466}
H.~{Dai}, Y.~{Huang}, Y.~{Xu}, C.~{Li}, B.~{Wang}, and L.~{Yang},
  ``Energy-efficient resource allocation for energy harvesting-based
  device-to-device communication,'' \emph{IEEE Trans. Veh. Technol.}, vol.~68,
  no.~1, pp. 509--524, Jan. 2019.

\bibitem{7872433}
Z.~Zhou, C.~Gao, C.~Xu, T.~Chen, D.~Zhang, and S.~Mumtaz, ``Energy-efficient
  stable matching for resource allocation in energy harvesting-based
  device-to-device communications,'' \emph{IEEE Access}, vol.~5, pp.
  15\,184--15\,196, 2017.

\bibitem{abramowitz1965handbook}
M.~Abramowitz and I.~A. Stegun, \emph{Handbook of mathematical functions: with
  formulas, graphs, and mathematical tables}.\hskip 1em plus 0.5em minus
  0.4em\relax Courier Corporation, 1965, vol.~55.

\bibitem{lee2015power}
N.~Lee, X.~Lin, J.~G. Andrews, and R.~W. Heath, ``Power control for {D2D}
  underlaid cellular networks: Modeling, algorithms, and analysis,'' \emph{IEEE
  J. Sel. Areas Commun.}, vol.~33, no.~1, pp. 1--13, Jan. 2015.

\bibitem{7797486}
A.~Memmi, Z.~Rezki, and M.~Alouini, ``Power control for {D2D} underlay cellular
  networks with channel uncertainty,'' \emph{IEEE Trans. wireless Commun.},
  vol.~16, no.~2, pp. 1330--1343, Feb. 2017.

\bibitem{6609136}
K.~{Huang}, ``Spatial throughput of mobile {A}d {H}oc networks powered by
  energy harvesting,'' \emph{IEEE Trans. Inf. Theory}, vol.~59, no.~11, pp.
  7597--7612, Nov. 2013.

\bibitem{chiu1996stochastic}
S.~N. Chiu, D.~Stoyan, W.~S. Kendall, and J.~Mecke, \emph{Stochastic geometry
  and its applications}.\hskip 1em plus 0.5em minus 0.4em\relax 2nd ed.
  Hoboken, NJ, USA: John Wiley \& Sons, 1996.

\bibitem{8333687}
Z.~{Behdad}, M.~{Mahdavi}, and N.~{Razmi}, ``A new relay policy in {RF} energy
  harvesting for {IoT} networks—a cooperative network approach,'' \emph{IEEE
  Internet of Things J.}, vol.~5, no.~4, pp. 2715--2728, Aug. 2018.

\bibitem{moltchanov2012distance}
D.~Moltchanov, ``Distance distributions in random networks,'' \emph{Ad Hoc
  Netw.}, vol.~10, no.~6, pp. 1146--1166, Mar. 2012.

\bibitem{nguyen2012stochastic}
T.~V. Nguyen and F.~Baccelli, ``A stochastic geometry model for cognitive radio
  networks,'' \emph{The Comput. J.}, vol.~55, no.~5, pp. 534--552, 2012.

\bibitem{7073589}
H.~{Sun}, M.~{Wildemeersch}, M.~{Sheng}, and T.~Q.~S. {Quek}, ``{D2D} enhanced
  heterogeneous cellular networks with dynamic {TDD},'' \emph{IEEE Trans.
  Wireless Commun.}, vol.~14, no.~8, pp. 4204--4218, Aug 2015.

\bibitem{hildebrand1987introduction}
F.~B. Hildebrand, \emph{Introduction to numerical analysis}.\hskip 1em plus
  0.5em minus 0.4em\relax Courier Corporation, 1987.

\end{thebibliography}

\end{document}